\documentclass[12pt]{article}

\usepackage[shortlabels]{enumitem}
\usepackage{amsmath} 
\usepackage{amssymb} 
\usepackage{mathrsfs} 
\usepackage{mathtools} 
\usepackage{amsthm}
\usepackage{thmtools}

\usepackage{appendix}
\usepackage{fullpage}
\usepackage{caption}
\usepackage{subcaption}
\usepackage{graphicx}
\usepackage[table,xcdraw]{xcolor}
\usepackage{wrapfig}
\usepackage{floatrow}
\usepackage{multirow}
\usepackage{makecell}
\usepackage{nicefrac}

\usepackage[ruled,vlined,linesnumbered]{algorithm2e}
\definecolor{forestgreen}{rgb}{0.13, 0.55, 0.13}

\SetCommentSty{mycommfont}

\definecolor{ForestGreen}{rgb}{0.0333,0.4451,0.0333}
\definecolor{DarkRed}{rgb}{0.65,0,0}
\definecolor{Red}{rgb}{1,0,0}
\usepackage[linktocpage=true,
pagebackref=true,colorlinks,
linkcolor=DarkRed,citecolor=ForestGreen,
bookmarks,bookmarksopen,bookmarksnumbered]{hyperref}

\usepackage{cleveref}
\usepackage{thm-restate}
\usepackage[normalem]{ulem} 

\newtheorem{theorem}{Theorem}[section]
\newtheorem{lemma}[theorem]{Lemma}
\newtheorem{claim}[theorem]{Claim}

\newtheorem{definition}[theorem]{Definition}

\newcommand*\mc[1]{\mathcal{#1}}
\DeclarePairedDelimiter\abs{\lvert}{\rvert}
\DeclarePairedDelimiter\norm{\lVert}{\rVert}%
\DeclarePairedDelimiter{\ceil}{\lceil}{\rceil}

\makeatletter
\let\oldabs\abs
\def\abs{\@ifstar{\oldabs}{\oldabs*}}
\let\oldnorm\norm
\def\norm{\@ifstar{\oldnorm}{\oldnorm*}}
\let\oldceil\ceil
\def\ceil{\@ifstar{\oldceil}{\oldceil*}}
\makeatother

\newcommand{\instance}{\emph{Steiner Path Aggregation Problem}}

\newif\ifcomments
\commentstrue

\ifcomments
\usepackage[colorinlistoftodos,prependcaption,textsize=tiny]{todonotes}

\newcommand{\dnote}[1]{\todo[linecolor=red,backgroundcolor=yellow!25,bordercolor=red]{\textbf{DH:~}#1}}
\newcommand{\dnotein}[1]{\todo[linecolor=red,backgroundcolor=yellow!25,bordercolor=red,inline]{\textbf{DH:~}#1}}

\newcommand{\dqtodo}[1]{\todo[linecolor=red,backgroundcolor=yellow!25,bordercolor=red]{\textbf{DQC:~}#1}}
\newcommand{\dqtodoin}[1]{\todo[linecolor=red,backgroundcolor=yellow!25,bordercolor=red,inline]{\textbf{DQC:~}#1}}

\newcommand{\tcr}[1]{{\color{red} #1}}
\newcommand{\rnote}[1]{\tcr{Ravi: #1}}
\newcommand{\rtodo}[1]{\todo[linecolor=red,backgroundcolor=teal!25,bordercolor=red]{\textbf{Ravi:~}#1}}

\else 

\newcommand{\dnote}[1]{}
\newcommand{\dnotein}[1]{}
\newcommand{\dqtodo}[1]{}
\newcommand{\dqtodoin}[1]{}
\newcommand{\rnote}[1]{}
\newcommand{\rtodo}[1]{}

\fi 

\usepackage{authblk}

\date{}
\title{The Steiner Path Aggregation Problem}

\author[1]{Da Qi Chen}
\author[2]{Daniel Hathcock\thanks{Corresponding author: \texttt{dhathcoc@andrew.cmu.edu}. Supported in part by NSF Graduate Research Fellowship grant DGE-2140739.}}
\author[3]{D Ellis Hershkowitz\thanks{Supported in part by NSF grant CCF-2403236.}}
\author[2]{R. Ravi\thanks{This material is based upon work supported in part by the Air Force Office of Scientific Research under award number FA9550-23-1-0031.}}

\affil[1]{University of Virginia, Charlottesville, VA}

\affil[2]{Carnegie Mellon University, Pittsburgh, PA}

\affil[3]{Brown University, Providence, RI}

\begin{document}

\maketitle

\begin{abstract}
In the \instance, our goal is to aggregate paths in a directed network into a single arborescence without significantly disrupting the paths. In particular, we are given a directed multigraph with colored arcs, a root, and $k$ terminals, each of which has a monochromatic path to the root. Our goal is to find an arborescence in which every terminal has a path to the root, and its path does not switch colors too many times. We give an efficient algorithm that finds such a solution with at most $2\log_{\nicefrac{4}{3}}k$ color switches. Up to constant factors this is the best possible universal bound, as there are graphs requiring at least $\log_2 k$ color switches. 
\end{abstract}


\section{Introduction}

In the \instance, we are given a directed multigraph $G = (V, A)$ with root $r \in V$ and a set of $k$ \emph{terminals} $R \subseteq V \setminus \{r\}$. The arcs are colored such that each of the terminals has a monochromatic path to the root, which we call its \emph{proposed path}. 
The goal is to find a subtree in which every terminal has a unique directed path to the root (not necessarily its proposed path), called a Steiner in-arborescence, and such that no terminal pays a high \emph{switching cost}. That is, the number of times each unique terminal-to-root path switches colors is low.

While a subgraph of switching cost 0 can be found by taking all arcs in $G$, such a solution may be expensive to construct (in terms of number of arcs) and does not provide a simple way of routing to the root from each terminal. This is why we seek to aggregate the terminals' proposed paths by choosing a Steiner in-arborescence: it is cheap in that it is a minimal arc set in which every terminal has a path to the root, and the terminal-to-root paths are unique, enabling simple routing. In this work, we study how to construct such an in-arborescence while bounding the switching cost.


Observe that if $R = V \setminus \{r\}$ and the underlying graph is already an arborescence after removing parallel arcs, the problem can be solved using the \emph{heavy path decomposition} of Sleator and Tarjan~\cite{SleatorT81}. This is a decomposition of the arcs of the underlying arborescence into arc-disjoint ``heavy paths'' such that every node-to-root path intersects at most $\lceil \log_2 |V| \rceil$ of these heavy paths. An arc is called heavy if more than half of the descendants of its parent node descend from the arc, and otherwise is called light. A heavy path is then defined as a maximal contiguous path of heavy arcs along with a single parent light arc. In particular, every arc lies in exactly one of the heavy paths. To solve the \instance, we may compute this decomposition, then for each heavy path we select the corresponding colored arcs from the proposed path of the lowest node in the heavy path. See \Cref{fig:heavy-path}. With this choice of subgraph, the switching cost of a terminal is precisely the number of heavy paths it uses on its path to the root, so this ensures a switching cost of at most $\lceil \log_2 |V| \rceil$.

\begin{figure}[h]
    \centering
    \includegraphics[scale=0.75]{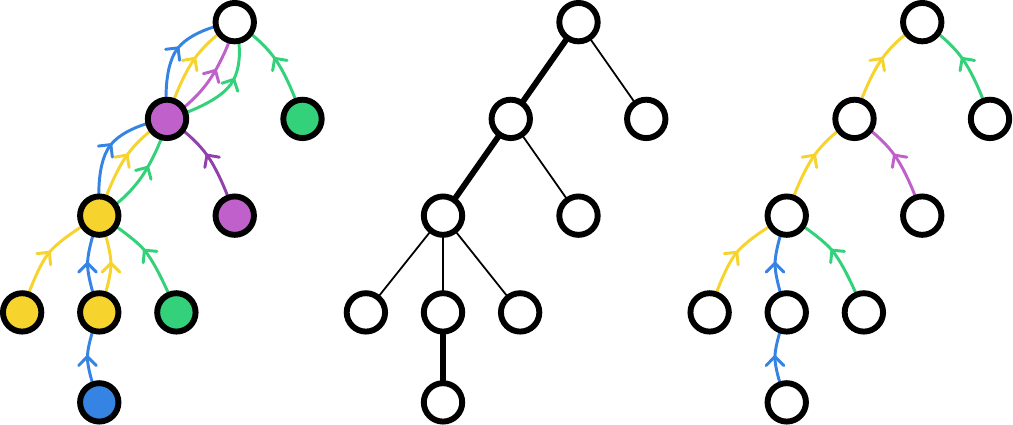}
    \caption{The input instance is pictured on the left. Observe that each node has a monochromatic ``preferred path'' to the root (indicated here by the color of the node). The underlying network after removing parallel arcs is pictured in the middle, with the heavy edges shown in bold. Each heavy path consists of a path of heavy edges and the parent light edge (if any). The right shows the subgraph chosen corresponding to the heavy path decomposition. Observe the switching cost is just 1, but a more naive choice could have incurred a switching cost of up to 3.}
    \label{fig:heavy-path}
\end{figure}

Even in the special case of trees, one cannot guarantee a switching cost less than $\log_2(\nicefrac{|V|}{2})$. For example, if the underlying graph is a complete binary tree with terminals $R = V \setminus \{r\}$, and all the proposed paths are made arc disjoint by adding parallel arcs and coloring them distinctly, then any solution will leave some node with switching cost at least $\log_2 (\nicefrac{|V|}{2})$ (this is easily seen by induction on the height of the tree). Thus, a natural question arises. 
\begin{quote}\centering
    \textit{Can we extend the $O(\log |V|)$ upper bound from trees to general graphs,\\ matching the simple $\Omega(\log |V|)$ lower bound on trees?}
\end{quote}


\paragraph{Our Contributions} We answer the above question affirmatively, giving an algorithm to find a low switching cost network in the general setting, whose guarantee matches (up to a constant) the $\log_2 |V|$ given by the heavy path decomposition for arborescences.

\begin{restatable}{theorem}{pathAggregation}\label{thm:path-agg}
    Let $G$ be an $n$-vertex directed multigraph with a root $r$ and $k$ terminals $R \subseteq V \setminus \{r\}$. Furthermore, suppose the arcs are colored such that every terminal $v \in R$ has a monochromatic dipath $P_v$ from $v$ to $r$. Then $G$ has an $r$-arborescence $T$ containing $R$ such that every terminal-to-root path in $T$ switches colors at most $2\log_{\nicefrac{4}{3}} k$ times.
\end{restatable}

Observe that our result gives a decomposition theorem for trees similar to the heavy path decomposition (if we relax the $\log_2 n$ to $O(\log n)$). For a tree $G$, construct the directed multigraph $G'$ on the same vertex set by picking a unique color for each non-root vertex $v$ and adding a dipath $P_v$ of that color corresponding to the node-to-root path of $v$ in $G$. Now applying \Cref{thm:path-agg} with terminals $R = V \setminus \{r\}$, the arborescence $T$ gives a coloring of $G$, and contiguous arcs of the same color form paths. In particular, every node-to-root path intersects at most $2\log_{\nicefrac{4}{3}} n$ of these paths, as in the heavy path decomposition. 

Thus, our result can be seen as extending the power of heavy path decompositions from trees---which has had a wide range of uses in both theoretical and practical settings---to general graphs. 
For example, Sleator and Tarjan originally used the heavy path decomposition to analyze the amortized cost of operations on their link/cut tree data structure~\cite{SleatorT81}. Chekuri, Hajiaghayi, Kortsarz, and Salavatipour use the heavy path decomposition as a step in showing the existence of junction trees for the non-uniform Buy-at-Bulk problem~\cite{CHKS10}. Naor, Umboh, and Williamson use the heavy path decomposition to give a tight $O(\log n)$-competitive algorithm for the Online Weighted Tree Augmentation problem~\cite{NaorUW19}. As another example, Haeupler, Hershkowitz, and Wajc used it to generalize broadcast algorithms from paths to trees \cite{haeupler2021near}.

\paragraph{Problem Difficulty} While the algorithm and analysis for the heavy path decomposition is quite simple, the same approach does not work in our more general setting. The algorithm for heavy path decomposition uses the size of subtrees rooted at a node to decide how to split the tree into paths. No such subtree structure exists in a general directed multigraph, so it is not clear how to proceed. Attempting to build the network iteratively from the terminals can result in circular dependencies, as in \Cref{fig:greedy-cycle}. These complex dependencies yield more interesting solutions than in the arborescence case. For example, a low switching cost network may include multiple disjoint segments of a single terminal's preferred path. 
\begin{figure}[h]
    \centering
    \includegraphics[scale=0.7]{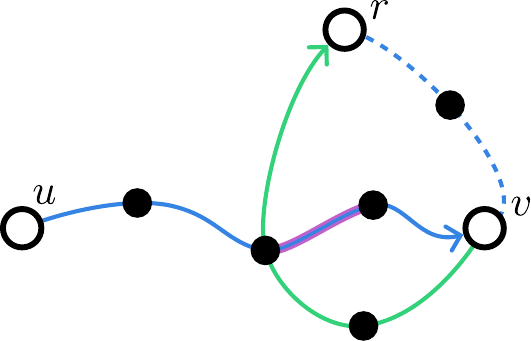}
    \caption{An example of the circular dependencies that may arise while greedily choosing arcs from the proposed paths of two nodes $u$ and $v$. An algorithm might select the arcs of some prefix of terminal $u$'s proposed path (blue). If the prefix ends at some other terminal $v$, then we may switch over to $v$'s proposed path (green). But now $v$'s proposed path could intersect the already chosen prefix of $u$'s proposed path, creating a cycle. This necessitates either dropping some arcs (i.e. the one highlighted in purple) or switching to a different proposed path before they intersect, both incurring some additional switching cost.}
    \label{fig:greedy-cycle}
\end{figure}

\paragraph{Algorithm Intuition} At a high level, our algorithm maintains a set of active nodes, initialized to the set of all terminals, and proceeds in rounds. In each round, the active nodes iteratively extend the solution by adding arcs along their proposed paths, maintaining vertex-disjointedness from the arcs added by other active nodes. Once no more arcs can be added in this way, we designate some active nodes to become inactive by allowing them to add one more arc from their proposed path, thus joining up with the proposed path of some other active node. In each round, a constant fraction of the active nodes becomes inactive, and any one node accumulates only a constant extra switching cost. So the total switching in the end is $O(\log k)$. 

Care must be taken to avoid circular dependencies, and to ensure the solution remains an arborescence. In the process of an active node $v$ adding arcs, it may intersect arcs in the solution added from a previous round (from a node which is now inactive), ``cutting through'' this path in the solution. One arc must be removed from the cut path to maintain that no node has out-degree more than 1, however this may change the component that certain nodes belong to in the partial solution. 


This approach is similar at a high level to the Matching Based Augmentation (MBA) framework for connectivity problems (see surveys \cite{Ravi06} and \cite[\S 7]{GUPTA11}), in which a partial solution of connected components is maintained. In each step, a low-cost augmentation is made to merge the components, reducing their number by, say, a constant factor resulting in a logarithmic approximation guarantee. The main distinction between MBA and our approach is that the circular dependencies mentioned above may cause the components of the partial solution to change unpredictably. A piece of one component may leave that component and join another at any point in the algorithm. Moreover, while MBA typically bounds the cost of each augmentation in terms of the optimal cost OPT, we give an absolute bound on the additional switching cost incurred by each iteration, independent of the instance.

\section{Preliminaries and Notation}

In this section we include some definitions and useful lemmas for the remainder of the paper. 

Given a directed multigraph $G = (V, A)$, an \emph{arborescence} (or $r$-arborescence) $T$ is a subgraph for which every node in $T$ has a unique directed path to a root $r$. A \emph{branching} in $G$ is a collection of vertex-disjoint arborescences (with different roots). 

\begin{definition}[$\alpha$-Switching Arborescence]
If the arcs of the graph $G$ are colored and a root $r$ is specified along with a terminal set $R$, then we call an $r$-arborescence \emph{$\alpha$-switching} if it contains $R$ and for every terminal $v \in R$ the path from $v$ to $r$ switches colors at most $\alpha$ times. 
\end{definition}

\begin{definition}[\instance]
An instance of the \instance\ is a directed multigraph $G$ with root $r$ and terminals $R \subseteq V$ in which the arcs are colored so that every terminal has a monochromatic directed path to the root. 
\end{definition}

The following simple lemma will be useful in defining our algorithm. 
\begin{lemma}\label{lem:3-color}
Let $H$ be an undirected graph in which every connected component $S$ has at most $\abs{V(S)}$ edges. Then $H$ is 3-colorable, and a 3-coloring can be found in polynomial time. 
\end{lemma}
\begin{proof}
Each component contains at most a single cycle, so the coloring can be found greedily (e.g. via a breadth-first search). 
\end{proof}

\section{Path Aggregation Algorithm}

We now describe the algorithm used to prove our main theorem, restated here in different terms.

\begin{theorem}\label{thm:main-restated}
Any instance $G$ of the \instance\ with $k$ terminals has a poly-time computable $2\log_{\nicefrac{4}{3}}k$-switching arborescence. 
\end{theorem}

\paragraph{Algorithm Overview} The algorithm maintains a partial solution as a branching $B$. Initially $B$ is empty, and the algorithm terminates when the terminals are all connected to $r$ in $B$. 

The algorithm also maintains a subset of the terminals called \textit{active nodes} $S$. In particular, each arborescence in $B$ that contains a terminal will contain a single active node representative. Each active node $v \in S$ maintains a prefix $P^A_v$ of its monochromatic path $P_v$, which we call its \textit{active path}. Initially, $S = R$, and the active paths are just the length-0 prefixes $P_v^A = \{v\}$. We will ensure that at every step of the algorithm, every terminal has a path in $B$ to the active path of its representative (i.e., the representative of its arborescence in $B$). 

See \Cref{alg:path-agg} for pseudo-code of the algorithm, and \Cref{fig:algorithm} for an example of one iteration. It proceeds iteratively by merging the arborescences in $B$:
\begin{enumerate}
    \item (\Cref{step:maxprefix}) First, it finds a set of \emph{maximal disjoint prefixes} for the active nodes. This is a set $\mc P$ of paths $P_v'$ for each active node $v \in S$ satisfying: (i) $P_v'$ is a prefix of $v$'s monochromatic path $P_v$; (ii) $P_v'$ extends $v$'s active path $P_v^A$; and (iii) the $P_v'$ are maximally vertex-disjoint, i.e. extending any $P_v'$ by one more arc would intersect some other $P_u'$. 
    
    This is done greedily: for each $v \in S$ (in arbitrary order), extend its active path $P_v^A$ along $P_v$ maximally while maintaining that it is vertex disjoint from all other $P_u^A$ and $P_u'$. 

    \item (\Cref{step:ext1,step:ext2,step:ext3}) Next, we merge arborescences by extending some of the paths $P_v'$ by a single arc, allowing them to ``join onto'' another path $P_u'$. 
    
    To select which paths to extend, we define an auxiliary \emph{dependency graph} $H$ with vertex set $V(H) = S$, and edge $uv \in H$ if extending $P_u'$ by one arc along $P_u$ intersects $P_v'$. If any $v \in S$ has $P_v' = P_v$ (i.e., it reaches the root), then $v$ is removed from $H$. Now, we find a 3-coloring of $H$, and let $S_H \subseteq S$ be one of the largest color classes. These are the nodes whose paths get extended. 

    \item (\Cref{step:update1,step:update2,step:update3,step:update4}) Finally, we update $B$ by adding all arcs from the paths in $\mc P$ (including the additional arc for each $v \in S_H$). We update the set of active nodes $S$ by removing $S_H$ from $S$, and for each still-active node $v \in S$, we update its active path $P_v^A \gets P_v'$. 

    We must also remove any arc from $B$ which leaves a vertex in any of the newly added paths from $\mc P$. The intuition of this step is that if a path $P_v'$ from $\mc P$ cuts through some path that was previously in $B$, the prefix of the cut path joins onto $P_v'$, while the suffix remains in $B$ (and will go on to join some other active path). 
\end{enumerate} 


\begin{algorithm}[h]	
    \caption{\texttt{Path Aggregation}}
    \label{alg:path-agg}
	\DontPrintSemicolon
	\SetKwInOut{Input}{input}\SetKwInOut{Output}{output}
	\Input{Digraph $G$ rooted at $r$ with $k$ terminals $R \subseteq V(G)$ and monochromatic paths $P_v$ for every $v\in R$.}
	\Output{$2\log_{\nicefrac{4}{3}} k$-switching arborescence $T$.}
    \BlankLine

    Initialize the representatives $S\gets R$ and the branching $B \gets \emptyset$.\;
    For each $v \in S$, initialize active paths $P_v^A \gets \{v\}$, the length-0 prefix of the path $P_v$.\;
    
    \While{there is a terminal without a path to $r$ in $B$}{
        Find the set $\mc P$ of \emph{maximal disjoint prefixes} $P_v'$ for each $v \in S$ greedily. \label{step:maxprefix}\;
        
        Let $H$ be the \emph{dependency graph} for $\mc P$.\label{step:ext1}\;
        
        Find a $3$-coloring of $H$, and let $S_H$ be one of the largest cardinality color classes. \label{step:ext2}\;
        
        For every $v\in S_H$, update $P_v'\in \mathcal{P}$ by extending $P_v'$ with one more arc along $P_v$. \label{step:ext3}\;
        
        Let $A_{\mc P}$ be the union of the set of arcs in the paths in $\mathcal{P}$.\label{step:update1}\;
        
        Let $B'$ be the arcs from $B$, minus those which exit any vertex in a path from $\mc{P}$. \label{step:update2}\;
        \textbf{Update:} $B \gets A_{\mc P} \cup B'$. \label{step:update3}\;
        \textbf{Update:} $S\gets S\setminus S_H$, and $P_v^A \gets P_v'$ for each $v \in S$.\label{step:update4}\;
    }
    
    \Return $B$.
\end{algorithm}

\begin{figure}[h]
\begin{tabular}{c@{\hskip 2cm}c}
  \includegraphics[scale=0.7]{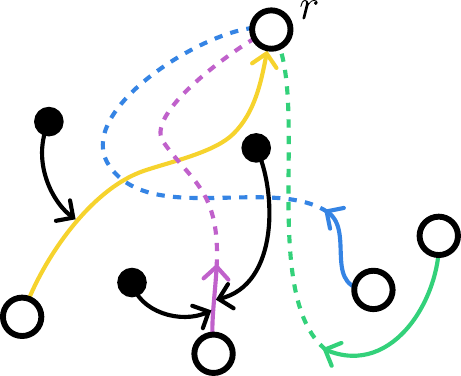} &   \includegraphics[scale=0.7]{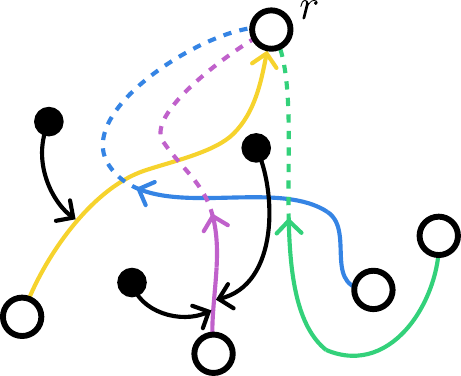} \\
(a) & (b) \\[0.5cm]
 \includegraphics[scale=0.7]{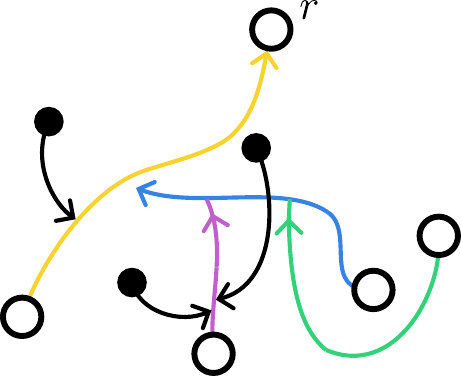} &   \includegraphics[scale=0.7]{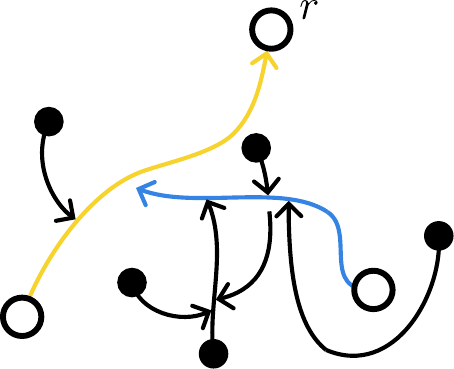} \\
(c) & (d)
\end{tabular}
\caption{An example iteration of \Cref{alg:path-agg}. (a) shows the state of the algorithm at the beginning of the iteration. Some nodes are inactive, shown in black, while white (non-root) nodes are active. Their active path is a solid color, and is a prefix of their dashed monochromatic path. (b) shows the \emph{maximal disjoint prefixes} $P'_v$ extending the active paths after \Cref{step:maxprefix}. (c) shows the result of \Cref{step:ext3}, extending the paths of those nodes in $S_H$ by a single arc (purple and green here). Finally (d) gives the result of the iteration. Note the nodes of $S_H$ became inactive, and also a previously inactive path was cut through by the new blue active path, so an arc was removed.}
\label{fig:algorithm}
\end{figure}

\begin{proof}[Proof of \Cref{thm:main-restated}]
To show the validity of the algorithm and bound the total number of rounds, we first prove the following.

\begin{claim}\label{claim}
At the end of the $i$-th iteration, the following conditions hold:
\begin{enumerate}
    \item $B$ induces a union of vertex-disjoint arborescences. \label{c1}
    \item Each arborescence $T'$ induced by the edges of $B$ is either a singleton non-terminal, or otherwise it contains exactly one active representative $v \in S$ (i.e. $|T'\cap S|=1$) and its active path $P_v^A$. \label{c2}
    \item If $v$ is the representative of $T'$, then any terminal $u\in V(T')$ can reach $V(P_v^A)$ in $T'$ with at most $2i$ color switches. \label{c3}
    \item The size of $S$ is reduced by at least $\frac{3}{4}$ after applying \Cref{step:update4}.\label{c4}
\end{enumerate}
\end{claim} 

\begin{proof}
We use induction to prove the above claims. Note that initially $B=\emptyset$, $S=R$, and each active path is $P_v^A = \{v\}$. Thus, $B$ induces a collection of isolated vertices and each terminal is its own representative and active path, satisfying conditions \ref{c1}, \ref{c2}, and \ref{c3} of the claim.
	
Let $B_i \subseteq A$, $S_i\subseteq R$, and $(P_{v}^A)_i$ for $v \in S_i$ denote the state of $B$, $S$, and the active paths $P_v^A$, respectively, at the end of iteration $i$. Assume that they satisfy conditions \ref{c1}, \ref{c2}, and \ref{c3}. We study the (weakly) connected components of $B_{i + 1} := (A_{\mc P})_{i +1} \cup B'_{i+1}$ created at the end of iteration $i+1$. First, $(A_{\mc P})_{i+1}$ consists of the paths $P'_v$ for $v \in S_i$ which are disjoint except for some which have been extended by a single arc to meet an unextended $P'_u$. Those that are unextended are precisely those that remain active, becoming $(P_u^A)_{i+1}$. Hence, each component of $(A_{\mc P})_{i+1}$ is either a singleton node in $V(G) \setminus S_i$, or an arborescence containing a single representative $v \in S_{i+1}$ and its active path $(P_v^A)_{i+1}$. We call the latter type of component an \emph{active component} of $(A_{\mc P})_{i+1}$. Observe that the final node of $(P_v^A)_{i+1}$ is the root of the arborescence represented by $v$, and hence has out-degree 0 in $(A_{\mc P})_{i+1}$.
	
To form $B_{i + 1}$, the algorithm removes any arc from $B_i$ leaving an active component of $(A_{\mc P})_{i+1}$ to create $B'_{i + 1}$, then adds the arcs from $B'_{i+1}$ to $(A_{\mc P})_{i+1}$. Since $(A_{\mc P})_{i+1}$ and $B_{i}$ (and thus also $B'_{i+1}$) have all out-degrees at most 1, then \Cref{step:update2} ensures the out-degrees of $B_{i+1}$ are also all at most 1. Therefore, to show that $B_{i + 1}$ induces a branching (condition \ref{c1}), it suffices to show that every component of $B_{i+1}$ has a vertex with out-degree 0. We show that each node $u \in V(G)$ has a path in $B_{i + 1}$ to an out-degree 0 node. 

First, suppose that $u$ lies in some active component of $(A_{\mc P})_{i + 1}$. The active component contains an active path $(P_v^A)_{i + 1}$ with a final node $f$. Since $u$ has a path to $f$ in the active component, the same path exists in $B_{i + 1}$. Moreover, $f$ has out-degree 0 in $(A_{\mc P})_{i+1}$, so by \Cref{step:update2} it still has out-degree 0 in $B_{i + 1}$. 

Second, suppose instead that $u$ is in a non-singleton component of $B_i$. From the inductive hypothesis, $u$ is contained in some arborescence $T'$ of $B_i$, and there is a path $L_u$ from $u$ to the active path $(P_v^A)_i$ of $T'$. We know that since $v \in S_{i}$, then in iteration $i+1$ we find a maximal disjoint prefix $P'_v$ of $P_v$, and $P'_v$ contains $(P_v^A)_i$. Now there are two cases (see \Cref{fig:cases}): 
\begin{enumerate}[I.]
    \item If $L_u$ has no arc leaving some vertex from any path in $\mc{P}$ (i.e., \Cref{step:update2} doesn't remove arcs from $L_u$), then all of $L_u$ is in $B'_{i+1}$, and therefore there is a path from $u$ to $P'_v$ in $B_{i+1}$.
    
    \item Otherwise, $L_u$ has an arc leaving some paths $P'_w \in \mc{P}$. Consider the earliest node $u'$ along $L_u$ at which this occurs. Then the prefix $u \to u'$ of $L_u$ is contained in $B'_{i+1}$, and $u'$ is in some $P'_w$. \label{case2-pathcross}
\end{enumerate}

\begin{figure}[h]
\begin{tabular}{c@{\hskip 2cm}c}
  \includegraphics[scale=0.7]{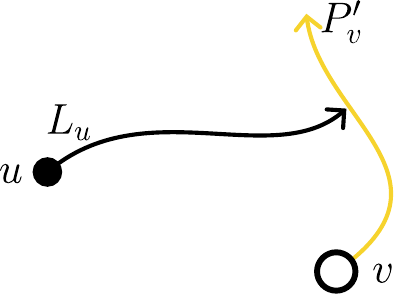} &   \includegraphics[scale=0.7]{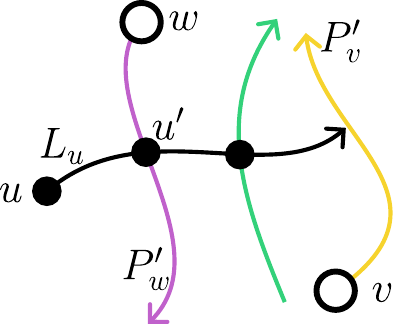} \\
Case I. & Case II.
\end{tabular}
\caption{An illustration of the two cases from the proof of \Cref{claim}.}
\label{fig:cases}
\end{figure}

In either case, $B_{i+1}$ contains a path from $u$ to some path in $\mc{P}$, and every such path is in an active component of $(A_{\mc{P}})_{i + 1}$. We already argued above that all such components contain an active path $P_v^A$, and the final node of that active path has out-degree 0 in $B_{i + 1}$. 

Finally, the only other possibility is that $u$ is neither in an active component of $(A_{\mc P})_{i + 1}$ nor a non-singleton component of $B_i$. Not being in an active component implies that $u$ is in a singleton component of $(A_{\mc P})_{i + 1}$, and hence $u$ itself has out-degree 0 in $B_{i + 1}$. This completes the proof of condition \ref{c1}. 

Condition \ref{c2} immediately follows: we showed that if $u$ lies in either an active component of $(A_{\mc P})_{i + 1}$ or a non-singleton component of $B_i$, then the arborescence in $B_{i + 1}$ containing $u$ has an active path $P_v^A$. No such component can have two representatives $v, v'$, or else it would contain two active paths and have two out-degree 0 nodes. This cannot happen in an arborescence. Otherwise, $u$ is in neither an active component of $(A_{\mc P})_{i + 1}$ nor a non-singleton component of $B_i$. But then $u \in V(G) \setminus S_i$, and hence $u$ must be a singleton non-terminal in $B_{i + 1}$ (or else we would contradict the induction hypothesis for $B_i$). 

To prove condition \ref{c3}, consider any terminal $u \in R$. Let its arborescence in $B_{i+1}$ be $T'$ with representative $v$. If $u$ is already in some active component of $(A_{\mc P})_{i + 1}$, then it switches colors at most once to get to its active path. Otherwise, it was in a non-singleton component of $B_i$. Observe that the path $L_u$ from the argument above is in $B_i$, so by induction switches colors at most $2i$ times. In Case \ref{case2-pathcross}, we may switch color an additional time at node $u'$. In either case, $L_u$ goes to some $P'_w \in \mc{P}$. If $P'_w$ is not the active path $(P^A_v)_{i+1}$ of $T'$, then it was extended by one arc in \Cref{step:ext3}, so it may switch color one further time to reach $(P^A_v)_{i+1}$. In total, the path $u \to (P^A_v)_{i+1}$ has at most $2i + 2$ color switches. 

Finally, condition \ref{c4} follows because the dependency graph $H$ is 3-colorable. By construction, $\mc{P}$ is a collection of paths where every pair is vertex-disjoint. When constructing $H$, for any vertex $v\in S_i$ where $P_v'\neq P_v$, the next vertex along $P_v$ hits at most one other path $P_u'$. We imagine directing the resulting edge $(v, u)$ in $H$ from $v$ to $u$, so each vertex in $H$ has out-degree at most 1. Therefore, any connected component $C$ of $H$ has at most $\abs{V(C)}$ edges. These graphs are 3-colorable by \Cref{lem:3-color}.

A largest color class $S_H$ in the 3-coloring is chosen to become inactive. $H$ has vertices $S_i$, possibly minus one vertex $v$ whose maximal disjoint prefix $P_v'$ reaches the root (note there is at most one such vertex by the vertex disjointness of the maximal disjoint prefixes). 
\[
    \abs{S_{i + 1}} \leq \abs{S_i} - \ceil{\frac{\abs{S_i} - 1}{3}} \leq \frac{3}{4} \cdot \abs{S_i}. \qedhere
\]
\end{proof}

Condition \ref{c4} implies that the algorithm runs for at most $\log_{\nicefrac{4}{3}} k$ iterations. This along with conditions \ref{c1}, \ref{c2}, and \ref{c3} together imply that $B$ connects all terminals to $r$ after this many iterations, and at the end it forms a $2\log_{\nicefrac{4}{3}} k$-switching arborescence for $G$. 

To complete the proof, we observe that \Cref{alg:path-agg} runs in polynomial time. This follows since the total number of iterations is $O(\log k) \leq O(\log n$), and each individual step clearly runs in polynomial time (the 3-coloring of \Cref{step:ext2} is polynomial time by \Cref{lem:3-color}).
\end{proof}

\bibliographystyle{alpha}\bibliography{references}

\end{document}